\documentclass[12pt]{amsart}

\usepackage{amssymb,amsthm,amsmath,enumerate}
\usepackage[numbers,sort&compress]{natbib}
\usepackage{color}
\usepackage{graphicx}
\usepackage{amssymb,amsthm,amsmath}
\usepackage[numbers,sort&compress]{natbib}
\usepackage{color}
\usepackage{graphicx}

\newcommand{\R}{\mathbb{R}}
\newcommand{\Z}{\mathbb{Z}}
\newcommand{\T}{\mathbb{T}}
\newcommand{\Q}{\mathbb{Q}}

\def\be{\begin{equation}}
\def\ee{\end{equation}}
\def\br{\begin{eqnarray}}
\def\er{\end{eqnarray}}

\theoremstyle{plain}
\newtheorem{theorem}{Theorem}[section]
\newtheorem{corollary}[theorem]{Corollary}
\newtheorem{lemma}[theorem]{Lemma}

\theoremstyle{definition}
\newtheorem{definition}[theorem]{Definition}

\newtheorem{remark}[theorem]{Remark}

\title[]{Sharp palindromic criterion for semi-uniform dynamical localization. }

\author{Svetlana Jitomirskaya}
\address[Svetlana Jitomirskaya]{ Department of Mathematics, University of California, Berkeley, California 94720, USA}
\email{sjitomi@berkeley.edu}

\author{Wencai Liu}
\address[Wencai Liu]{Department of Mathematics, Texas A\&M University, College Station, TX 77843-3368, USA}\email{liuwencai1226@gmail.com; wencail@tamu.edu}
\author{Lufang Mi}

\address[Lufang Mi]{College of Science, Shandong University of Aeronautics, Binzhou, 256600, China }\email{milufang@126.com}
\begin{document}

 \begin{abstract}We develop a sharp palindromic
   argument for general 1D operators, that proves absence
   of semi-uniform localization in the regime of exponential symmetry-based
   resonances. This
  provides the first examples of operators with dynamical
   localization but no SULE/SUDL, as well as with nearly uniform
   distribution of centers of localization in absence of SULE. For the almost Mathieu operators, this also leads to a sharp
   arithmetic criterion for semi-uniformity of dynamical localization in the Diophantine case. 
\end{abstract}

 \maketitle

\section{Introduction and main results}
A corollary of the celebrated RAGE theorem, e.g. \cite{cycon}, states
that dynamical localization, (i.e. boundedness in $t$ of
$\|xe^{-itH}\delta_0\|$) implies pure point spectrum of a self-adjoint
operator $H:\ell^2(X^d)\to \ell^2(X^d),\; X\in\{\Z,\R\}$. The fact that
the converse is not true has been known since \cite{djls,djlsprl}. Indeed,
even Anderson localization, that is a complete set of normalized eigenfunctions
$\phi_s$ satisfying
\begin{equation}\label{and}
|\phi_s(n)|  \leq C_se^{-\gamma |n-n_s|}, 
\end{equation}
for some $C_s<\infty, n_s\in X^d,$ can coexist with nontrivial (even almost ballistic) dynamics
\cite{djls}. Other examples of this sort, including in a physically
relevant context, were studied in \cite{jss,js,gk}. The reason why
eigenfunction estimates as in (\ref{and}) may not translate into
dynamical bounds  is the lack of control over the constants $C_s$
required in (\ref{and}). Indeed, a uniform bound on $C_s,$ so called
{\it uniform} localization, would immediately do the job, but in  most
interesting situations, save for a couple of notable exceptions \cite{mar,dg,jk, jk2,chu,kerz,csz,kps}, it
does not hold \cite{djls,j}. However, ~\cite{djls} was not all bad
news, because it introduced the sufficient level of control: for any $\varepsilon>0,$ 
\begin{equation}\label{su1}
C_s\leq C_\varepsilon e^{\varepsilon |n_s|}, 
\end{equation}

dubbed semi-uniform localization (SULE). Indeed, it was shown in
\cite{djls} that  for operators with simple spectrum SULE is equivalent to semi-uniform dynamical 
localization (SUDL):
\begin{equation}\label{su2}
\sup_t| (\delta_n, e^{-itH}\delta_m)|\leq  C_\varepsilon e^{\varepsilon |m|}e^{-\gamma |n-m|},
\end{equation}
The standard definitions for other types of dynamical localization are
\begin{definition}[DL]
	We say $H$ on $\ell^2(\Z^d)$ has dynamical localization if for any exponentially decaying  $\phi$ and $p$, 
$\sum_{n\in\Z^d} |n|^p|(e^{-itH} \phi,\delta_n)|^2$ is bounded.
\end{definition}

\begin{definition}[EDL]
	We say $H$ on $\ell^2(\Z^d)$ has exponential dynamical localization if for any exponentially decaying  $\phi$ and $p$,  there exist $C$ and $\gamma$ such that 
$$\sup_{t} |(e^{-itH} \phi,\delta_n)|\leq Ce^{-\gamma |n|}.$$
\end{definition}
\begin{remark}\cite{djls,tchecmp01}
As mentioned, in one dimension (or, more generally, for operators with simple spectrum), SULE is equivalent to SUDL. 
    Clearly, SULE implies EDL and EDL implies DL.
\end{remark}
Since then, upgrading mere localization to  SULE has become an important pathway to proving
dynamical localization, e.g. \cite{gd,jl3,gj,g,ds01,gk1,gkjems13,gyz23}. It has remained
unclear though whether dynamical localization can hold without SULE,
especially in a non-artificial model. Here we resolve this by
\begin{theorem}\label{exist}
There exist explicit operators with DL/EDL but no SULE/SUDL.
\end{theorem}
{\bf Remark.} It should be noted that SULE/SUDL implies semi-stability
of dynamical localization with respect to fast decaying perturbations
\cite{djls}: power-logarithmic in $t$ bound on the growth of
$\|xe^{-itH}\delta_0\|$ is stable. In contrast, even exponential dynamical localization in the form 
\begin{equation}\label{su}
\sup_t| (\delta_n, e^{-itH}\delta_m)|\leq C_m e^{-\gamma |n-m|},
\end{equation}
without control on the growth of $C_m$ does not imply similar dynamical stability for anything but the
compactly supported perturbations. It is an interesting question
whether there is indeed (and what level of) dynamical instability in
our examples. A similar issue is the dimensional characteristics of
spectral measures of localized perturbations.
%time of the
%moments of the position operator) 

Our result does not come from an artificial construction but from a
new delocalization criterion that holds in high generality. It does
not even require ergodicity and  includes  the entire class of almost
periodic potentials.
%Let $f:\Z\to\R$ be an almost periodic function with hull $U_f$ and Haar
%measure $\mu$. 
Consider discrete Schr\"odinger operator
$$(H_{g}u)(n)=u(n+1)+u(n-1)+ g(n)u(n)$$ on  $   \ell^2(\mathbb{Z}),$
%where  $g\in U_f,$ is called almost periodic. 
where $g\in  \ell^{\infty}(\Z).$
Let $R$ be the reflection map $R:\ell^{\infty}(\Z)\to
\ell^{\infty}(\Z)$, $Rg(n)=g(-n),$ % for $g\in U_f$,
and $T$ be the shift, $Tg(n)=g(n+1).$ %for $g\in U$.
Define $d(g_1,g_2)=||g_1-g_2||_{\ell^{\infty}(\Z)}$. 

Define one-step transfer matrices
$$
A_n(E):=\begin{pmatrix}
E-g(n)&-1\\
1&0
\end{pmatrix}
$$
and the multi-step transfer matrices by
$$
A_{n,k}(E):=\begin{cases}
A_{n-1}(E)A_{n-2}(E)\cdots A_k(E),&k< n;\\
A_{n,k}=I,&k=n;\\
A_{k,n}^{-1}(E),&k>n.
\end{cases}
$$

Define \begin{equation}\label{Lambda}\mathcal{L}(E) :=
\limsup_{n\to\infty}\sup_{ k\in\Z} \frac{\ln \| A_{n+k,k}(E)\|}{n}.\end{equation} 
 Clearly, for
 bounded $V,$ $\mathcal{L}(E)<\infty$, for every $E,$ and for potentials $g(n)=V(T^nx)$ where $T$ is uniquely ergodic and $V$ continuous,
 $\mathcal{L}(E)$ coincides with the ($x$-independent) Lyapunov exponent.

Let $\delta(g)\in[0,\infty)$ be defined by $$\delta(g):=\limsup_{|n|\to\infty}
-\frac{\ln d(RT^{n}g,g)}{ |n|}.$$ % Let $L(E)$ be the Lyapunov exponent of $H_g.$
Then we have
\begin{theorem}\label{crit}
Operator $H_g$ has
\begin{enumerate} 
 \item no point spectrum/dynamical localization on $\{E:\delta(g)>\mathcal{L}(E)\}$;
\item  no SULE/SUDL  if $\delta(g)>0.$
\end{enumerate} 
\end{theorem}

Part (1) of Theorem \ref{crit}  is the sharp abstract version of the
palindromic argument of \cite{js94}.   It was proved in \cite{jljems}
for the particular setting of one-frequency quasi-periodic Schr\"odinger operators, however the
proof of our abstract statement is exactly the same, so we omit it
here, and instead concentrate on the proof of part (2). 

Theorem
\ref{crit} leads to an immediate corollary for quasiperiodic operators.
Assume $v$  is a function from $\T\to \R$.  
Operators
$$(H_{v,\alpha,\theta}u)(n):=u(n+1)+u(n-1)+ v(2\pi(\theta+n\alpha))u(n),$$
where  $\alpha\in \R\backslash \Q$ and $\theta\in \R$, are called
quasiperiodic. Let $||x||=\text{ dist }(x,\Z)$. Define an
arithmetic parameter $\delta(\alpha,\theta)\in[0,\infty)$ by
\begin{equation}\label{G.delta}
\delta(\alpha,\theta):=\limsup_{k\to\infty} -\frac{\ln ||2\theta+ n\alpha||_{\R/\Z}}{|n|}
\end{equation}

Part (2) of Theorem \ref{crit} implies

\begin{corollary}\label{absence1}
	Suppose $v:\mathbb{R}/\mathbb{Z}\to \mathbb{R}$ is  even and  H\"older continuous,  and $\delta(\alpha,\theta)>0$. Then
	$H_{v,\alpha,\theta}$ does not have SULE/SUDL.
\end{corollary}

Theorem \ref{crit} is sharp since

\begin{theorem}\label{sharp}

\begin{enumerate} 
\item For any $0<b< \mathcal{L}<\infty$, there exists 
%an almost periodic
%  function $f$ and $g\in U_f$
$g\in \ell^{\infty}(\Z)$
  with $\delta(g)=b$ and $\mathcal{L}(E)\ge \mathcal{L}$ for all $E$, such that $H_g$ has point spectrum (even Anderson
  localization) but no SULE/SUDL.
\item There exists 
%$f$ and $g\in U_f$
$g\in \ell^{\infty}(\Z)$
  with $\delta(g)=0$ such that $H_g$ has SULE/SUDL.
\end{enumerate} 
\end{theorem}

The examples for Theorem \ref{sharp} will be\footnote{Of course!} from
the almost Mathieu family, operator  
 \begin{equation}\label{AMO}
 (H_{\lambda,\alpha,\theta}u)(n)=u({n+1})+u({n-1})+ 2\lambda v(\theta+n\alpha)u(n),  \text{ with }  v(\theta)=\cos2\pi \theta,
 \end{equation}
where $\lambda>1$ is the coupling, $\alpha $ is the frequency, and
$\theta $ is the phase. %It is, of course, an almost periodic operator
%with $U= \R/\Z$ and $\mu$ the Lebesgue measure. 

%Combined with some past results, 
Corollary \ref{absence1} leads to a
sharp trichotomy for the almost Mathieu family with Diophantine
frequency (see (\ref{Diokt}) for the definition) depending on the properties of the phase. 

We have the following trichotomy

\begin{theorem}\label{amcrit}
For Diophantine $\alpha,$ the operator $H_{\lambda,\alpha,\theta}$ with $\lambda>1$ has
\begin{enumerate} 

\item no point spectrum if $\delta(\alpha,\theta)>\ln \lambda$;
\item EDL/DL but no
  SULE/SUDL  if $0<\delta(\alpha,\theta)<\ln \lambda$;
\item   SULE/SUDL/EDL if $\delta(\alpha,\theta)=0. $
\end{enumerate} 
\end{theorem}

For comparison, we list

\begin{theorem}\cite[Theorem 1.1]{jljems}\label{thmal}
For Diophantine $\alpha,$ operator $H_{\lambda,\alpha,\theta}$ with $\lambda>1$ has
\begin{enumerate} 

\item no point spectrum if $\delta(\alpha,\theta)>\ln \lambda$ ;
\item  Anderson localization if $\delta(\alpha,\theta)<\ln \lambda.$\footnote{A new proof was also recently given in \cite{gyz24}}
\end{enumerate} 

\end{theorem} 
Thus part (1) of Theorem \ref{amcrit} was already proved in
\cite{jljems}  and is only included here
for completeness. Previously, exponential dynamical localization was established
for the almost Mathieu operator for a.e. $\theta$ in \cite{gj}, after
earlier results \cite{jl3,g}, thus only a measure-theoretic
version. The {\it arithmetic} results on dynamical localization
(2) were not known previously. However, our main focus here is the
result on {\it absence} of SULE.

Another issue we explore in this paper is uniformity of the
distribution of centers of localization. 
It is natural to take as centers of localization of exponentially
decaying eigenfunctions their global maxima. Indeed,  %let $m_s$ be a global maximum
%of the exponentially localized eigenfunction $\phi_s$,
%$s=1,2,\cdots$.  
we have
\begin{lemma}\label{lemax}
	
	Suppose $H$ has SULE. Let $m_s$ be any global maximum of $\phi_s$, $s=1,2,\cdots,$ and $n_s$ be given by \eqref{and} and \eqref{su1} in the definition of SULE.
	Assume $||\phi_s||=1$,$s=1,2,\cdots$. Then for any $\varepsilon>0$,
	\begin{equation*}
	|m_s-n_s|\leq \varepsilon |m_s|, 	|m_s-n_s|\leq \varepsilon |n_s|,
	\end{equation*}
	for large $n_s$.
\end{lemma}
\begin{proof}
	It immediately follows from the definition of SULE and normalization.
\end{proof}
Given an orthonormal basis of eigenfunctions, the distribution of their centers of
localization $m_s$ is governed by the quantities
$$d_{+}:=\limsup_{L\to \infty} \frac{\#\{s:|m_s|\leq L\}}{2L+1}$$
and
$$ d_{-}:=\liminf_{L\to \infty} \frac{\#\{s:|m_s|\leq L\}}{2L+1}.\footnote{with obvious modifications in the
  multidimensional case}$$
One may notice that SULE 
does not require {\it any} information on the centers of localization $m_s$,
but it is misleading: it turns out that, combined with orthogonality
and completeness, SULE actually implies uniform distribution of $m_s:$
for any choice of centers of localization $m_s$, as long as
(\ref{and})  and (\ref{su1}) are satisfied, we have that they have uniform density:
\begin{theorem}\label{thmsule}\cite[Theorem 7.1]{djls}
Suppose $H$ has SULE. Then
\begin{equation*}
  d_+=d_-=1.
\end{equation*}
\end{theorem}

The main intuition behind examples with dynamical delocalization/lack of SULE in both
\cite{djls,js} is that eigenfunctions become more and more extended,
at least along certain scales, so that one can choose most $m_s$ within
certain scales to be
concentrated near the origin, making it possible to show that $d_{+}=\infty.$
Also, using some techniques of \cite{jz}  one can show that,
in the  example in \cite{djls},
it is possible to choose  most $m_s$ within
certain scales to be appropriately away from the origin, so that
$d_{-}=0.$

This leads to a question whether such extended nature of eigenstates,
manifested in strong non-uniformity of centers of localization,
is a necessary feature of absence of SULE. Here we show it is not the
case. 

We have
\begin{theorem}\label{thmden}
Let $v=2\lambda\cos2\pi\theta$. Suppose $\alpha$ is Diophantine and $\ln \lambda>\delta(\alpha,\theta)$. Then
\begin{equation*}
d_+\leq 1+\frac{\delta(\alpha,\theta)}{2\ln \lambda}
\end{equation*}
\begin{equation*}
d_-\geq 1-\frac{\delta(\alpha,\theta)}{\ln \lambda}.
\end{equation*}
\end{theorem}
Moreover, this can be strengthened to the statement about ``almost maxima''
(which would e.g. allow the choice of the origin, leading to
$d_{+}=\infty$ in the examples of \cite{djls,js}).
Fix any positive constant $K$.
 Pick  any  $\tilde{m}_s\in \Z$  such that
 \begin{equation}\label{newden}
   \tilde{m}_s\geq \frac{1}{K} \max_{k}|\phi_s(k)|,  s=1,2,\cdots.
 \end{equation}
  Define the upper (lower) density of the almost maxima as
$$\tilde{d}_{+}=\limsup_{L\to \infty} \frac{\#\{s:|\tilde{m}_s|\leq L\}}{2L+1},$$
and
$$ \tilde{d}_{-}=\liminf_{L\to \infty} \frac{\#\{s:|\tilde{m}_s|\leq L\}}{2L+1}.$$

\begin{theorem}\label{thmdennew}
Let $v=2\lambda\cos2\pi\theta$. Suppose $\alpha$ is Diophantine and $\ln \lambda>\delta(\alpha,\theta)$. Then
\begin{equation*}
\tilde{d}_+\leq 1+\frac{\delta(\alpha,\theta)}{2\ln \lambda}
\end{equation*}
\begin{equation*}
\tilde{d}_-\geq 1-\frac{\delta(\alpha,\theta)}{\ln \lambda}.
\end{equation*}
\end{theorem}

It is an interesting question whether Theorems
\ref{thmden} and \ref{thmdennew} are sharp, and also whether there exist
operators  with no SULE yet $\tilde{d}_+=\tilde{d}_-=1.$ 

\begin{remark}\label{recon}
	We have the following facts:
\begin{itemize}
	\item Theorem \ref{exist} is an immediate corollary of  (2) of Theorem \ref{amcrit}, which  follows immediately from Corollary
	\ref{absence1}.
	\item Theorem \ref{sharp} follows from Theorem \ref{amcrit}.
	\item Theorem \ref{thmdennew} follows from Theorem \ref{thmden} by straightforward modifications.
	\item Part 1 of Theorem \ref{crit} has the same proof as part 1 of Theorem \ref{amcrit}, which is  proved in \cite{jljems}.
	\item  SULE  is equivalent to SUDL and  implies EDL, which, in turn, implies DL.
\end{itemize}

Therefore, we only need to prove
\begin{enumerate}
 \item Part (2) of Theorem \ref{crit}, see Section \ref{seccrit};
		\item Theorem \ref{thmden}, see Section \ref{secthmden};
	\item EDL statement in part (2)  and SULE statement in part (3) of Theorem \ref{amcrit}, see Section \ref{secamcrit}.

\end{enumerate}
 
\end{remark}

\section{Proof of part (2) of Theorem \ref{crit}}\label{seccrit}

Since $H_g$ is bounded, we can assume $E\in[-K,K]$.
    Assume 
 $H_{g}\varphi= E\varphi$.  It is easy to check that  for any $m\in \Z$ and $k\in \Z^+$, one has
 \begin{equation}\label{gb11}
 \left(\begin{array}{c}
 \varphi(m+k) \\
 \varphi(m+k-1)                                                                                        \end{array}\right)=
 \prod_{j=k-1}^{0}
 \begin{pmatrix}
 E-g(m+j)& -1 \\
 1& 0
 \end{pmatrix}
 \left(\begin{array}{c}
 \varphi(m) \\
 \varphi(m-1)                                                                                        \end{array}\right).
 \end{equation}
   Let 
  \begin{equation}\label{gb12}
  B=\max_{n\in\Z,E\in[-K,K]} \ln \left\|\begin{pmatrix}
  E-g(n)& -1 \\
  1& 0
  \end{pmatrix}
  \right\|.
  \end{equation}
By \eqref{gb11} and \eqref{gb12}, we have  that
 for   any $k,m$,
 \begin{equation}\label{G.new17}
\left\| 
 \left(\begin{array}{c}
 \varphi(m) \\
 \varphi(m-1)                                                                                        \end{array}\right)\right\| e^{-B|k|}\leq 
 \left\|   \left(\begin{array}{c}
                                                                                            \varphi(k+m) \\
                                                                                           \varphi(k+m-1)                                                                                        \end{array}\right)
                                                                                          \right\|\leq e^{B|k|}\left\|
 \left(\begin{array}{c}
                                                                                            \varphi(m) \\
                                                                                          \varphi(m-1)                                                                                        \end{array}\right)\right\|.
 \end{equation}

 Frequency $\alpha$ is called Diophantine if there exist  $\kappa>0$ and $\tau>0$ such
 that  for $k\neq 0$,
 \begin{equation}\label{Diokt}
 ||k\alpha||_{\R/\Z}\geq \frac{\tau}{|k|^{\kappa}}.
 \end{equation}

 Let $u$ be a  normalized eigensolution, i.e.,
$H_{g}u=Eu$ and
\begin{equation*}
 || u||= \sum_n|u(n)|^2=1.
\end{equation*}

Given a sequence $\{k_i\}_{i=1}^{\infty},$ we let $u_i(n)= u(k_i-n)$, $V(n)=g(n)$ and $V_i(n)=g(k_i-n)$.
 Since $\delta(g)>0,$ there exists $\epsilon>0$ and  $\{k_i\}_{i=1}^{\infty}$ with $\lim| k_i|=\infty$ such that
  for all $n\in\Z$,
 \begin{equation}\label{Eqp}
    |V(n)-V_i(n)|\leq Ce^{-\epsilon |k_i|}.
 \end{equation}

 % By the assumption, one has there exists $\{k_i\}_{i=1}^{\infty}$ with $\lim k_i=\infty$ such that
 %  for all $n\in\Z$,
 % \begin{equation}\label{Eqpnew}
 %    |V(n)-V_i(n)|\leq Ce^{- \epsilon|k_i|}.
 % \end{equation}
Let

 \begin{equation*}
   \Phi(n)=\left(\begin{array}{cc}
                   u(n) \\ u(n-1)
                 \end{array}
   \right);
   \Phi_i(n)=\left(\begin{array}{cc}
                   u_i(n) \\ u_i(n-1)
                 \end{array}
   \right).
 \end{equation*}

We will use
\begin{theorem}\label{thmphase}\cite[ Theorem 4.2]{jljems}
The following statements hold for large $k_i$,
\begin{itemize}
  \item $k_i$ is even. Let $m_i=\frac{k_i}{2}$.
   There exists $\iota\in\{-1,1\}$  such that
 \begin{equation*}
    ||\Phi(m_i)+\iota\Phi_i(m_i)||\leq C e^{-\frac{1}{2}\epsilon|k_i|}.
 \end{equation*}
  \item  $k_i$ is odd. Let $\tilde{m}_i=\frac{k_i-1}{2}$.
  There   exists $\iota\in\{-1,1\}$  such that
 \begin{equation*}%\label{Phicase2}
    ||\Phi(\tilde{m}_i+1)+\iota\Phi_i(\tilde{m}_i+1)||\leq C e^{-\frac{1}{2}\epsilon|k_i|}.
 \end{equation*}
\end{itemize}
\end{theorem}
\begin{proof}[ \bf Proof of part (2) of Theorem \ref{crit}]
 Without loss of generality, assume $k_i>0$.
 Let
 \begin{equation}\label{bg1}
     \epsilon'=\frac{\epsilon}{20B}.
 \end{equation}
 Suppose $H_{g}$ has SULE. By  Theorem \ref{thmsule} and Lemma \ref{lemax},   for large $i$, there exists a normalized eigensolution $u$  of
$H_{g}u=Eu$ such that
there exists
\begin{equation}\label{bg2}
    m\in [k_i/2-3\epsilon^\prime k_i,k_i/2-2\epsilon^\prime
k_i]
\end{equation}
satisfying
\begin{equation}\label{nov131}
  |u(m)|=\max_{n\in \Z}|u(n)|,
\end{equation}
and 
there exists  $t>0$ such that  for  
 all $n\in\Z$,
\begin{equation}\label{nov133}
|u(n)|\leq C_{\varepsilon}e^{\varepsilon |m|-t|n-m|} .
\end{equation}
For convenience, we drop the dependence on $\varepsilon$ for the rest
of this proof.
By \eqref{Eqp} and Theorem \ref{thmphase}, one has that for each $i$
there exists $\iota\in\{-1,1\}$  such that
\begin{equation}\label{nov135}
    ||\Phi(m_i)+\iota\Phi_i(m_i)||\leq C e^{-\frac{1}{2} \epsilon k_i}.
 \end{equation}
where $m_i=\frac{k_i}{2}$ if  $k_i$ is even and $m_i=\frac{k_i-1}{2}$
if $k_i$ is odd.

% \begin{itemize}
%   \item[Case 1] $k_i$ is even. Let $m_i=\frac{k_i}{2}$.
%    There exists $\iota\in\{-1,1\}$  such that
%  \begin{equation}\label{nov135}
%     ||\Phi(m_i)+\iota\Phi_i(m_i)||\leq C e^{-\frac{1}{2} \epsilon k_i}.
%  \end{equation}
%   \item[Case 2]  $k_i$ is odd. Let $\tilde{m}_i=\frac{k_i-1}{2}$.
%   There   exists $\iota\in\{-1,1\}$  such that
%  \begin{equation*}%\label{Phicase2}
%     ||\Phi(\tilde{m}_i+1)+\iota\Phi_i(\tilde{m}_i+1)||\leq C e^{-\frac{1}{2}\epsilon k_i}.
%  \end{equation*}
% \end{itemize}
We assume $k_i$ is even; the odd case is very similar and we omit the details.

 Let $T_i^1 $ and $T_i^2$ be the transfer matrices  with the potentials  $V$ and $V_i$ respectively, taking $\Phi(m_i),\Phi_i(m_i)$ to $\Phi(m),\Phi_i(m)$.
 By \eqref{gb12}, (\ref{Eqp}) and telescoping, one has that
  \begin{equation}\label{nov136}
    ||T_i^1||, ||T_i^2||\leq C e^{ B |m-m_i|},
  \end{equation}
  and
  \begin{equation}\label{nov137}
    ||T_i^1-T_i^2||\leq C e^{ B|m-m_i|-2\epsilon m_i}.
  \end{equation}
By   \eqref{nov135}, \eqref{nov136} and \eqref{nov137}, we have
\begin{eqnarray}
    %||\Phi(0)+\sigma\Phi(2m_i-1)|| &=&   \\
    ||\Phi(m)+\iota\Phi_i(m)|| &=& ||T_i^1\Phi(m_i)+\iota T_i^2 \Phi_i(m_i)|| \nonumber\\
    &=& ||T_i^1\Phi(m_i)+\iota T_i^1 \Phi_i(m_i)-\iota T_i^1 \Phi_i(m_i)+\iota T_i^2 \Phi_i(m_i)|| \nonumber \\
     &\leq &  || T_i^1||||\Phi(m_i)+\iota\Phi_i(m_i)||+ ||T_i^1-T_i^2|| ||\Phi_i(m_i) ||\nonumber\\
      &\leq &C  e^{ B|m-m_i|-  \epsilon m_i} + Ce^{ B|m-m_i|-2\epsilon m_i}\nonumber\\
       &\leq &C e^{ -\frac{4}{5}\epsilon m_i} \label{laadd},
  \end{eqnarray}
  where the last inequality holds by \eqref{bg1} and \eqref{bg2}.
  %Let   $m^\prime=2m_i-m=k_i-m$. 
  By \eqref{nov133}, we have
  \begin{equation*}
   | u(m)|\geq e^{-\varepsilon k_i}
  \end{equation*}
  and
  \begin{equation*}
   | u( k_i-m)|+ | u(k_i-m+ 1)|\leq e^{-t\epsilon^\prime k_i}.
  \end{equation*}
  This contradicts  \eqref{laadd} and thus completes the proof.
\end{proof}

\section {Proof of Theorem \ref{thmden}}\label{secthmden}
In this section, we assume $\ln \lambda>\delta(\alpha,\theta).$
By the definition of $\delta(\alpha,\theta)$, for large $n$
\begin{equation}\label{Defdelta}
  ||2\theta+n\alpha||_{\R/\Z}\geq e^{-(\delta+\varepsilon)|n|}.
\end{equation}

Let $\phi_s$ be a normalized eigensolution of $H\varphi=E\varphi$. Let $m_s$ be  a global maximum of $\phi_s$:
\begin{equation*}
    |\phi_s(m_s)|=\max_{n\in\Z}|\phi_s(n)|.
\end{equation*}
If $\sin\pi(2\theta+L_0\alpha)=0$, the eigenfunction must be even or odd around $L_0/2$ and therefore $|\phi_s(m_s)|=|\phi_s(L_0-m_s)|.$ If $|sin\pi(2\theta+L_0\alpha)|$ is the smallest of $|sin\pi(2\theta+L\alpha)|$ of $L$ in a large interval, points around $L_0-m_s$ are where the eigenfunction rises, and we informally say that $L_0-m_s$ is the resonance. 
We will use the following decay estimate that gives an upper bound on the eigenfunction even in the most resonant situations.
\begin{theorem}\cite[Theorem 2.1]{jklmrl}\label{Keytheorem}
 Let $\lambda>1$, $\alpha$ Diophantine, $\theta\in\R$,
 $\ell\in\Z$, and $\ell^\prime=|\ell-m_s|$.
 Let $k_0\in[-2\ell^{\prime},2\ell^{\prime}]$  be such that
 \begin{equation}
  |\sin\pi(2\theta+\alpha (2m_s+k_0))|=\min_{|x|\leq 2\ell^{\prime}}
  |\sin\pi(2\theta+\alpha (2m_s+x))|.
 \end{equation}
 Then  for large $\ell^{\prime}$ (depending on $\varepsilon$) we have
 \begin{itemize}
   \item if $\ell$ and $k_0+m_s$ are on different sides of $m_s$, that is $(\ell-m_s)k_0<0$, then
   \begin{equation}
 |\phi_{s}(\ell)|\leq
   e^{-(\ln\lambda -\varepsilon)|\ell-m_s|}
   |\phi_{s}(m_s)|.
 \end{equation}

   \item if $\ell$ and $k_0+m_s$ are on the same side of $m_s$, that is  $(\ell-m_s) k_0\geq0$, and  $
 |\sin\pi(2\theta+\alpha (2m_s+k_0))|\geq
    e^{-\eta |\ell-m_s|}$ for  some $\eta\in(0,\ln\lambda-\varepsilon)$,
 then
   \begin{equation}
 |\phi_{s}(\ell)|\leq
   e^{-(\ln\lambda -\eta-\varepsilon)|\ell-m_s|}
   |\phi_{s}(m_s)|.
 \end{equation}
 \end{itemize}

\end{theorem}
Denote by   $ U^{\varphi}(y) =\left(\begin{array}{c}
                             \varphi(y)\\
                            \varphi({y-1})
                          \end{array}\right).
                     $
The following Lemma establishes decay between the resonances. 
\begin{lemma}\cite[Lemma 3.4]{jljems}\label{Keylemmaapp}
%Let $\phi_s$ be the eigensolution $H\varphi=E\varphi$ with a global maximu $m_s$.
 % Let $\varphi$ be any solution of $H\varphi=E\varphi$ and define
%$r_{y}^{\varphi}=\max_{|\sigma|\leq 10 \gamma}|\varphi(y+\sigma k)|$.
Assume that  $k_0\in[-2C_1k,2C_1k]$  satisfies
 \begin{equation*}
  |\sin\pi(2\theta+2m_s\alpha+  k_0 \alpha)|=\min_{|x|\leq 2 C_1k}
  |\sin\pi(2\theta+2m_s\alpha+ x\alpha )|,
 \end{equation*}
 where $C_1\geq1$ is a   constant.
 Let $\gamma,\varepsilon$ be small positive constants.
%Let $r_{y}^{\varphi}=\max_{|\sigma|\leq 10\gamma}|\varphi(y+\sigma k)|$.
Let $y_1=m_s, y_2=m_s+k_0, y_3\in m_s+[-2C_1k,2C_1k]$.% and $|y_2|,|y_3|\geq k$.
Assume $y$ lies in  $[y_i,y_j]$ (i.e., $y\in [y_i,y_j]$)with  $|y_i-y_j|\geq k$ and $y_s\notin [y_i,y_j]$, $s\neq i,j$.
 Suppose     $|y_i|,|y_j|\leq C_1k$ and $|y-y_i|\geq 10\gamma k$, $|y-y_j|\geq 10\gamma k$.
  Then for large enough $k$,
  \begin{align}
     ||U^{\phi_s} (y)||\leq \max\{&||U^{\phi_s} (y_i)||\exp\{-(\ln \lambda- \varepsilon)(|y-y_i|-14\gamma k)\},\nonumber\\
   &||U^{\phi_s} (y_j)||\exp\{-(\ln\lambda- \varepsilon)(|y-y_j|-14\gamma k)\}\}. \label{second1}
  \end{align}

\end{lemma}
Lemma \ref{Keylemmaapp} immediately implies 
\begin{corollary}\label{keycor}
    Assume that  $k_0\in[-10C_1k,10C_1k]$  satisfies
 \begin{equation*}
  |\sin\pi(2\theta+k_0\alpha)|=\min_{|x|\leq 10 C_1k}
  |\sin\pi(2\theta+ x\alpha )|,
 \end{equation*}
 where $C_1\geq 3$ is a   constant. Assume $|m_s|\leq k$. Then
  for any $y$ with $|y|\leq C_1k$
   \begin{align}
     ||U^{\phi_s} (y)||\leq e^{\varepsilon k}\max\{&||U^{\phi_s} (m_s)||\exp\{-(\ln \lambda- \varepsilon)|y-m_s|\},\nonumber\\
   &||U^{\phi_s} (k_0-m_s)||\exp\{-(\ln\lambda- \varepsilon)|y-(k_0-m_s)|\}\}. \label{second}
  \end{align}
\end{corollary}
In particular, it implies that eigenfunction decays exponentially with distance to the nearest resonance.

By Theorem \ref{thmal}, $H_{\lambda,\alpha,\theta}$ has a complete eigenbasis 
$\{\phi_s\}_{s=1}^{\infty}$.  By orthogonality and completeness we have
\begin{equation}\label{completebasis1}
   \sum_{n\in\Z} |\phi_s(n)|^2=1,
\end{equation}
and
\begin{equation}\label{completebasis2}
      \sum_{s\in\Z^+} |\phi_s(n)|^2=1.
\end{equation}

\begin{lemma}\label{keylemo51}
     For any $J\subset [-CL,CL]\cap \Z$, one has that
     \begin{equation}\label{go53}
         \sum_{s\in \Z^+, m_s\in J\atop{n\in [-CL,CL]}} |\phi_s(n)|^2\leq C\varepsilon L+ 2\# J.
     \end{equation}
 \end{lemma}
 \begin{proof}
     Let $J_1$ be $C\varepsilon L$ neighborhood of  $J$.
          Let
  $L_0$   be such that
             \begin{equation*}%\label{x0}
            |\sin\pi(2\theta+L_0\alpha)|  = \min_{|x|\leq CL}|\sin\pi(2\theta+x\alpha)|.
             \end{equation*}

By Corollary \ref{keycor}, one has that for any $n$ with $n\in [-CL,CL]\backslash(J_1\cup (L_0-J_1))$  and $m_s\in J$, 
\begin{equation}\label{bgo51}
    |\phi_s(n)|\leq e^{-\varepsilon L}. 
\end{equation}
By \eqref{bgo51} and \eqref{completebasis2}, one has that
\begin{eqnarray*}
 \sum_{s\in \Z^+, m_s\in J\atop{n\in [-CL,CL]}} |\phi_s(n)|^2 &\leq &  e^{-\varepsilon L}+\sum_{s\in \Z^+, m_s\in J \atop {n\in J_1\cup (L_0-J_1)}}|\phi_s(n)|^2 \\
    &\leq&  e^{-\varepsilon L}+\sum_{s\in \Z^+  \atop {n\in J_1\cup (L_0-J_1)}}|\phi_s(n)|^2  \\
     &\leq&C\varepsilon L+ 2\# J.
\end{eqnarray*}

 \end{proof}
 
\begin{proof}[\bf Proof of the upper bounds of Theorem  \ref{thmden}]
	
	% Fix an eigenfunction $\phi_n$  and let $m_n$ be a global maximum. 
 Assume $|m_s|\leq L$, $L$ sufficiently large. By \eqref{Defdelta}, one has for any $|n|\geq CL,$ 
             \begin{eqnarray}
             ||2\theta+2m_s\alpha+n\alpha||_{\R/\Z} &\geq& e^{-(\delta+\varepsilon)|n+2m_s|} \\
             &\geq&  e^{-\eta |n-m_s|}
             \end{eqnarray}
             for some  $ \eta  $ with $\delta<\eta<\ln\lambda$.
             By Theorem \ref{Keytheorem}, we have  for any $n$ with $|n|\geq CL$,
             \begin{equation}\label{newdecay}
             |\phi_s(n)|\leq  e^{-(\ln\lambda-\eta-\varepsilon)|n|}.
             \end{equation}
             We will now estimate $|\phi_s(n)|$ for $(1+\varepsilon) L\leq |n|\leq CL$.
 \begin{lemma} \label{35}For $L>0$ large enough and $|m_s|\leq L$
      there exists a subset $J\subset \Z$ with
                \begin{equation}\label{10}
               \#J\leq \frac{\delta}{\ln \lambda} L+C\varepsilon L,
               \end{equation}
               such that  
               for any $n$ with $n\notin J$  and $(1+\varepsilon) L\leq |n|\leq CL$, 
               \begin{equation}\label{11}
               |\phi_s(n)|\leq e^{-\varepsilon L}.
               \end{equation}
               
 \end{lemma}

 {\bf Proof}.

We split the estimate of  $\phi_s(n)$  into several cases.
 Let
  $L_0$ (we can choose any one if $L_0$ is not unique) be such that
             \begin{equation*}%\label{x0}
            |\sin\pi(2\theta+L_0\alpha)|  = \min_{|x|\leq CL}|\sin\pi(2\theta+x\alpha)|.
             \end{equation*}
             We then have that $L_0-m_s$ is the resonance, so in estimating $\phi_s(n)$ we need to worry about the distances from $n$ to $m_s$ and $L_0-m_s.$
    Without loss of generality, assume   $L_0\geq 0$. By \eqref{Defdelta},
one has that
\begin{equation}\label{L0}
||2\theta+L_0\alpha||_{\R/\Z} \geq e^{-(\delta+\varepsilon) L_0}.
\end{equation}   
\\
%{\color{blue} For Lana: $L_0-m_s$ is the resonance, $m_s$ is the maximum}\\
             
     {\bf Case 1:}        $|m_s-L_0|\leq L$
     
      For any $n$ with $(1+\varepsilon)L\leq |n|\leq CL$  we have that $|n-(m_s+k_0)|>\varepsilon L,$ and therefore, applying Corollary \ref{keycor},
     we have
     \begin{equation}\label{g1}
     |\phi_s(n)|\leq   e^{-\varepsilon L}.
     \end{equation}

        {\bf Case 2:}      $|m_s-L_0|> L$ 
        and $m_s\leq \frac{1}{2}(1-\frac{\delta}{\ln\lambda})L_0-C\varepsilon L$
     
      In this case, we will show that  for any $n$ with $(1+\varepsilon)L\leq |n|\leq CL$,
     \begin{equation}\label{g3}
     |\phi_s(n)|\leq  e^{-\varepsilon L}.
     \end{equation}
     
  {\bf Case $2_1$:}    $|n-(L_0-m_s)|\geq C\varepsilon L$
  
  In this case,  applying  Corollary \ref{keycor},  we have
  
  \begin{equation}\label{g4}
  |\phi_s(n)|\leq e^{-\varepsilon L}.
  \end{equation}
  
  {\bf Case $2_2$:}    $|n-(L_0-m_s)|\leq C\varepsilon L$

   By \eqref{L0} and the fact that  $m_s\leq \frac{1}{2}(1-\frac{\delta}{\ln\lambda})L_0-C\varepsilon L$,
    we have
     \begin{equation}\label{g2}
   ||2\theta +L_0\alpha||_{\R/\Z}\geq e^{ -(\delta+\varepsilon)L_0}\geq e^{-(\ln\lambda-C\varepsilon) |n-m_s|}.
     \end{equation}
     
     Applying the second case of Theorem \ref{Keytheorem} with $k_0=L_0-2m_s$, $\ell=n$, one has that for any $n$ with $|n-(L_0-m_s)|\leq C\varepsilon L$,
       \begin{equation}\label{g6}
     |\phi_s(n)|\leq e^{-\varepsilon L}.
     \end{equation}
     
     \eqref{g3} follows from \eqref{g4} and \eqref{g6} immediately.
     \\
            {\bf Case 3:}      $|m_s-L_0|> L$ and $m_s> \frac{1}{2} (1-\frac{\delta}{\ln\lambda})L_0-C\varepsilon L$

             In this case, by the assumption that $L_0\geq 0 $, $|m_s|\leq L$  and $|m_s-L_0|> L$,  one has that
              \begin{equation*}
              m_s\leq L_0-L.
              \end{equation*}

              In this case we will need to exclude from our estimates the potential resonant points. Those form an interval that we will call $J$ or $J'$. 
            
           \par
                     {\bf Case $3_1$:}  $L_0\leq 2L$
              
                If $\frac{1}{2}(1-\frac{\delta}{\ln\lambda}-C\varepsilon)L_0> L_0-L$, let $J=\emptyset$.
                If $\frac{1}{2}(1-\frac{\delta}{\ln\lambda}-C\varepsilon)L_0\leq L_0-L$,
               let $I=[\frac{1}{2}(1-\frac{\delta}{\ln\lambda})L_0 , L_0-L]\subset [0,L]$, the set of values of $m_s$ that can lead to resonances,  and  $J_1=L_0-I$. 
               Let $J\subset \Z$ be the $C\varepsilon L$ neighbourhood of $J_1$.
                 Applying  Corollary \ref{keycor}  again, one has that
                 for any $n$ with $n\notin J$  and $|n|\geq (1+\varepsilon) L$, 
                 \begin{equation}\label{g7}
               |\phi_s(n)|\leq e^{-\varepsilon L}.
               \end{equation}
               
                 {\bf Case $3_2$:}  $L_0>2L$

               If $\frac{1}{2}(1-\frac{\delta}{\ln\lambda}-C\varepsilon)L_0>L$, let $ J'=\emptyset$. If $\frac{1}{2}(1-\frac{\delta}{\ln\lambda}-C\varepsilon)L_0\leq L$,
               let $I^\prime=[\frac{1}{2}(1-\frac{\delta}{\ln\lambda})L_0 , L]\subset [0,L]$,
               and $J_1^\prime=L_0-I^\prime$.  Let $J^\prime\subset \Z$ be the $C\varepsilon L$ neighbourhood of $J_1^\prime$.
               
               Applying  Corollary \ref{keycor}  again, one has
               for any $n$ with $n\notin J^\prime$  and $|n|\geq (1+\varepsilon)L$, 
               \begin{equation}\label{g8}
               |\phi_s(n)|\leq e^{-\varepsilon L}.
               \end{equation}
               
               It is easy to check that 
               \begin{equation}\label{g9}
               \#J\leq \frac{\delta}{\ln \lambda} L+C\varepsilon L,   \# J^\prime\leq \frac{\delta}{\ln \lambda} L+C\varepsilon L.
               \end{equation}
               
               %Putting all cases together, we prove that  there exists a subset $J\subset \Z$ such that 
                %\begin{equation}\label{10}
               %\#J\leq \frac{\delta}{\ln \lambda} L+C\varepsilon L,
               %\end{equation}
               %and   
              % for any $n$ with $n\notin J$  and $(1+\varepsilon) %L\leq |n|\leq C L$, 
              % \begin{equation}\label{11}
              % |\phi_s(n)|\leq e^{-\varepsilon L}.
              % \end{equation}
               \qed
              % Similarly, the arguments also work for $L_0<0$.

%and

%\begin{equation}\label{Gtwosides}
  %  |\phi_n(m)|\leq  \sup_{|k-L_0|\leq \varepsilon L}|\phi_n(k)|,
%\end{equation}
%if $|m-L_0|\leq \varepsilon L$.

%For $|k|> C  L$, by  (\ref{smalluniform}), we have
%\begin{eqnarray}
%% \nonumber to remove numbering (before each equation)
%  ||2(\theta+m_n\alpha)+k\alpha||_{\R/\Z} &\geq &  \sigma e^{-(\delta+\sigma)|k+2m_n|} \label{newsmall1}\\
%   &\geq& \sigma e^{-(\delta+\sigma_1)|k|} \label{newsmall},
%\end{eqnarray}
% where  $\delta+\sigma_1<\ln|\lambda|$( by letting $C$ is large enough).  %is  much  than $\sigma$ in (\ref{newsmall1}). We do not want to use too many notations so just keep the same.

%where $\sigma $ is a proper constant such that  $\delta+\sigma<\ln|\lambda|$ (by letting $C$ large).

Let $S=[-(1+\varepsilon)L, (1+\varepsilon)L]\cup J$.
By  Lemma \ref{35}, \eqref{newdecay},  and \eqref{completebasis1}, one has under the assumption that $|m_s|\leq L$, 
\begin{equation}\label{g15}
\sum_{n\in S} |\phi_s(n)|^2\geq 1-e^{-\varepsilon L}.
\end{equation}
By Lemma \ref{35},  (\ref{completebasis2}) and \eqref{g15}, one has
\begin{eqnarray*}
  2L+\frac{\delta}{\ln \lambda} L +C\varepsilon L &\geq &  \sum_{s\in\Z^{+},n\in S}|\phi_s(n)|^2 \\
    &\geq&  \sum_{s\in\Z^+,|m_s|\leq L,n\in S} |\phi_s(n)|^2  \\
    &\geq&  \sum_{s\in\Z^+,|m_s|\leq L} (1-e^{-\varepsilon L} )\\
     &\geq&(1-e^{-\varepsilon L} )\#\{s:|m_s|\leq L\}.
\end{eqnarray*}
 This implies
 \begin{equation}\label{g101}
    \limsup (2L+1)^{-1}\#\{s,|m_s|\leq L\}\leq 1+\frac{\delta}{2\ln \lambda}.
 \end{equation}
 \end{proof}
 
\begin{proof}[\bf Proof of the lower bounds of Theorem  \ref{thmden}]

 Assume $|m_s|\geq (1+\varepsilon)L$ and $ |n|\leq L$.

 {\bf Case 1:} $|m_s|\geq C L$.
 
 We will prove that there exists some $\eta>0$ such that for  all $n$ with $|n|\leq L$,
 \begin{equation}\label{g18}
 |\phi_s(n)|\leq e^{- \eta  |m_s|} . 
 \end{equation}
  Without loss of generality, assume $m_s\geq CL$.
  
  Let $L_0\in[-10 |m_s|,10|m_s|]$  be such that
 \begin{equation}\label{g19}
 |\sin\pi(2\theta+L_0\alpha) |=\min_{|x|\leq 10|m_s|}
 |\sin\pi(2\theta+x\alpha)|.
 \end{equation}

 If $L_0-2m_s\geq 0$ (namely $L_0-m_s\geq m_s$), applying     Corollary \ref{keycor}, one has that
 for any $n$ with $|n|\leq L$,
 \begin{equation*}
 |\phi_s(n)|\leq e^{-(\ln\lambda-\varepsilon)|n-m_s|}.
 \end{equation*}
 This implies \eqref{g18}.
 
Therefore, it suffices to consider the case $L_0-2m_s<0$.

{\bf Case $1_1$:} $|m_s|\geq C_0 L$ and  $|L_0-m_s|\geq  \frac{10}{C_0}m_s$.

Applying  Corollary \ref{keycor}, one has
for any $n$ with $|n|\leq L$,
\begin{equation}
|\phi_s(n)|\leq e^{-(\ln\lambda-\varepsilon)|n-(L_0-m_s)|} + e^{-(\ln\lambda-\varepsilon)|n-m_s|}.
\end{equation}
This implies \eqref{g18}.

{\bf Case $1_2$:} $|m_s|\geq C_0 L$ and  $|L_0-m_s|\leq\frac{10}{{C_0}}m_s$.

By \eqref{Defdelta} and \eqref{g19}, one has for some $\eta$ with $\delta<\eta<\ln 
\lambda$, 
 \begin{equation}\label{g20}
|\sin\pi(2\theta+L_0\alpha )| \geq e^{-\eta |n-m_s|},
\end{equation}
for all $n$  with $|n|\leq L$.

Applying Theorem \ref{Keytheorem}  with \eqref{g20}, we have for any $n$ with $|n|\leq L$,
\begin{equation}
|\phi_s(n)|\leq  e^{-(\ln\lambda-\eta-\varepsilon)|n-m_s|}.
\end{equation}
Since $|m_s|>(1+\varepsilon)L$, this implies \eqref{g18}.

 {\bf Case 2:} $(1+\varepsilon)L\leq |m_s|<C L$.
 
  Let
 $L_0$ (we can choose any one if $L_0$ is not unique) be such that
 \begin{equation*}%\label{x0}
 |\sin\pi(2\theta+L_0\alpha)|  = \min_{|x|\leq CL}|\sin\pi(2\theta+x\alpha)|.
 \end{equation*}
 
 {\bf Case $2_1$:}   $|m_s-L_0|\geq (1+\varepsilon)L$. Then $|n-m_s|>\varepsilon L$ and $|n-(m_s-L_0)|>\varepsilon L$.
 Applying  Corollary \ref{keycor}, 
 we have that for any $n$ with $ |n|\leq L$,
 \begin{equation}\label{g21}
 |\phi_s(n)|\leq  e^{-\varepsilon L}.
 \end{equation}
 
{\bf Case $2_2$:}   $|m_s-L_0|< (1+\varepsilon)L$.

We are going to prove that 
there exists a subset $J\subset [-CL,CL] \backslash [-L-\varepsilon L,L+\varepsilon L]$ such that 
\begin{equation}\label{12new}
\#J\leq \frac{\delta}{\ln \lambda} L+C\varepsilon L,
\end{equation}
and   
for any $m_s$ with $m_s\notin J$  and $   |n|\leq  L$, 
\begin{equation}\label{13new}
|\phi_s(n)|\leq e^{-\varepsilon L}.
\end{equation}

Without  loss of generality, assume $L_0\geq 0$. 
Then we must have $m_s\geq (1+\varepsilon)L$.
By   Corollary \ref{keycor} and Theorem \ref{Keytheorem} ,  we only need to show
there exists an appropriate size subset $J\subset [-CL,CL] \backslash [-L-\varepsilon L,L+\varepsilon L]$ such that  for $m_s\notin J$,
\begin{equation}\label{12newn}
|\sin\pi(2\theta+L_0\alpha )|\geq  e^{-( \ln \lambda- C\varepsilon)|L_0-2 m_s|}.
\end{equation}
By  \eqref{L0} this would follow from  
\begin{equation}\label{12newnew}e^{-(\delta +\varepsilon)L_0} \geq  e^{-( \ln \lambda- C\varepsilon)|L_0-2 m_s|}.
\end{equation}
 
In this case, $L_0-2m_s=L_0-m_s-m_s<0.$
 The estimate in \eqref{12newnew} becomes
\begin{equation}\label{12newnew1}
 (\delta +\varepsilon)L_0 \leq    ( \ln \lambda- C\varepsilon)(2m_s-L_0).
\end{equation}

{\bf Case $2_{21}$:} $L_0\geq 2L$

If   $L_0-L>\frac{\delta }{2 \ln\lambda} L_0+\frac{L_0}{2}
+C\varepsilon L$, let  $J=\emptyset$.
If $L_0-L\leq \frac{\delta }{2 \ln\lambda} L_0+\frac{L_0}{2}
+C\varepsilon L$, let $J$ be $\varepsilon L$ neighborhood of 
 $[L_0-L, \frac{\delta }{2 \ln\lambda} L_0+\frac{L_0}{2}]$.

 {\bf Case $2_{22}$:} $L_0\leq 2L$ 

If   $L>\frac{\delta }{2 \ln\lambda} L_0+\frac{L_0}{2}
+C\varepsilon L$, let  $J=\emptyset$.
If $L\leq \frac{\delta }{2 \ln\lambda} L_0+\frac{L_0}{2}
+C\varepsilon L$, let $J$ be $C\varepsilon L$ neighbor of 
 $[L, \frac{\delta }{2 \ln\lambda} L_0+\frac{L_0}{2}]$.

 Direct computations show that \eqref{12newnew1} and \eqref{12new} hold.
 
%for any $n$ with
% $ |n-(L_0-m_s)|\geq \varepsilon L $,
% one has that 
%\begin{equation}\label{newg21}
% |\phi_s(n)|\leq   e^{-\varepsilon L}.
% \end{equation}
% So we  only to consider $n$ with
% $ |n-(L_0-m_s)|\leq \varepsilon L $.

 Putting all cases together and by \eqref{g101}, we have
 \begin{align*}
\sum_{|n |\leq L, m_s\notin J \atop^{s\in\Z^+,|m_s|\geq(1+\varepsilon)L}} |\phi_s(n)|^2\ &=  \sum_{|n |\leq L, m_s\notin J \atop^{s\in\Z^+, |m_s|\geq C L}} |\phi_s(n)|^2+\sum_{|n|\leq L, m_s\notin J \atop^{s\in\Z^+,(1+\varepsilon)L\leq |m_s|\leq CL}} |\phi_s(n)|^2  \\
&\leq e^{-\varepsilon L} .
 \end{align*}
 %where  we use the upper bound in the last inequality.  

By \eqref{completebasis1}, (\ref{completebasis2}), \eqref{g101},\eqref{12newnew1} and \eqref{go53}, we have that
\begin{eqnarray*}
  2 L-2\frac{\delta}{\ln \lambda} L-C\varepsilon L&\leq &  \sum_{s\in\Z^{+},|n|\leq  L,m_s\notin J}|\phi_s(n)|^2 \\
    &=&  \sum_{|n|\leq L,m_s\notin J\atop^{s\in \Z^+,|m_s|\leq (1+\varepsilon)L}} |\phi_s(n)|^2 + \sum_{|n |\leq L,m_s\notin J\atop{s\in\Z^+,m_s|>(1+\varepsilon)L}} |\phi_s(n)|^2\\
    &\leq&  \sum_{|n |\leq L\atop{s\in\Z^+,|m_s|\leq (1+\varepsilon)L}} |\phi_s(n)|^2+e^{-\varepsilon L} \\
     &\leq&\#\{s:|m_s|\leq (1+\varepsilon)L\}+e^{-\varepsilon L}.
\end{eqnarray*}
Thus the following lower bound holds
 \begin{equation*}
    \liminf (2L+1)^{-1}\#\{s:|m_s|\leq L\}\geq 1 -\frac{\delta}{\ln \lambda}.
 \end{equation*}
\end{proof}

%\section{appendix}
\section{Proof of part (2) of Theorem \ref{amcrit}}\label{secamcrit}

\begin{theorem}\label{edl}
	There exist constants $C(\lambda,\delta)$, $C_1(\lambda,\delta)$ and $C_2(\lambda,\delta)$ such that for large enough $m_s$,
	\begin{equation}\label{gedl}
	|\phi_s(n)\phi_s(\ell)|\leq Ce^{-C_1|n-\ell|+C_2|n|}.
	\end{equation}
\end{theorem}
\begin{proof}
	%By changing $C$ in \eqref{gedl}, we can assume $m_s$ is large enough. 
 We split the proof into several cases.
 
{\bf Case 1}: $|n|\geq C |m_s|$ or $|\ell|\geq C |m_s|$

Without loss of generality, assume  $|n|\geq C |m_s|$.
By \eqref{Defdelta}, one has that for some $\eta$ with $\delta<\eta<\ln 
\lambda$, 
\begin{equation}\label{g20newb}
|\sin\pi(2\theta+(2m_s+n)\alpha )| \geq e^{-\eta |n-m_s|},
\end{equation}
for any $|n|\geq  C|m_s|$.
By
Theorem \ref{Keytheorem}, we have  for all $n$ with $|n|\geq  C |m_s|$,
\begin{equation}\label{gedl1}
|\phi_s(n)|\leq  e^{-(\ln\lambda-\eta-\varepsilon)|n|}.
\end{equation}
If $|\ell|\leq \frac{1}{100} |m_s|$, by \eqref{gedl1},\eqref{gedl} holds.

If $|\ell|\geq \frac{1}{100} |m_s|$, by  \eqref{gedl1}, one has 
\begin{equation}\label{ngedl2}
|\phi_s(\ell)|\leq  e^{-(\ln\lambda-\eta-\varepsilon)|\ell|}.
\end{equation}

By \eqref{gedl1} and \eqref{ngedl2},
  \eqref{gedl}  holds.

{\bf Case 2}: $ \frac{1}{C}  |m_s| \leq |n|\leq C|m_s|$, $|\ell|  \leq  C |m_s|$ 
or 
$ \frac{1}{C}  |m_s| \leq |\ell |\leq C|m_s|$, $|n|  \leq  C |m_s|$ 

In this case,  \eqref{gedl} holds immediately (noting that $C_2$ is much larger than $C_1$).

{\bf Case 3}: $|n|\leq \frac{1}{C}  |m_s|$, $|\ell|  \leq \frac{1}{C}|m_s|$

Let $L_0\in[-C|m_s|,C|m_s|]$  be such that
\begin{equation}
|\sin\pi(2\theta+L_0\alpha)|=\min_{|x|\leq  C|m_s|}
|\sin\pi(2\theta+\alpha  x))|.
\end{equation}

{\bf Case $3_1$}: $|L_0-m_s|\geq \frac{2}{C}  |m_s|.$ Then $|m_s-n|>|n|, |L_0-m_s-n|>|n|$ and similarly for $\ell.$

Thus, by  Corollary \ref{keycor}, one has that
\begin{equation}\label{newdecay22}
             |\phi_s(\ell)|\leq  e^{-(\ln\lambda -\varepsilon)|\ell|}, |\phi_s(n)|\leq  e^{-(\ln\lambda -\varepsilon)|n|}.
             \end{equation}
             This implies \eqref{gedl}.

  {\bf Case $3_2$}: $|L_0-m_s|\leq \frac{2}{C}  |m_s|$

In this case, 
by \eqref{Defdelta}, one has that for some $\eta$ with $\delta<\eta<\ln 
\lambda$, 
\begin{equation}\label{g20new}
|\sin\pi(2\theta+ L_0\alpha )| \geq e^{-\eta |n-m_s|},
\end{equation}
for any $|n|\leq  \frac{1}{C}|m_s|$.
 
By Theorem \ref{Keytheorem},  one has that
for all $n$ with $|n|\leq \frac{1}{C} |m_s|$,
\begin{equation}
|\phi_s(n)| \leq   e^{-(\ln\lambda -\varepsilon-\eta)|n-m_s|} \leq  e^{-\frac{1}{2}(\ln\lambda -\eta) |n|}.
\end{equation}
This implies \eqref{gedl}.

\end{proof}
\begin{proof}[\bf Proof of the EDL statement in part (2) of Theorem \ref{amcrit}]
By Theorem \ref{Keytheorem}, for small $m_s$ (the smallness depends on $\lambda,\alpha $, and $\theta$),
\begin{equation}\label{gbo6new}
   | \phi_s(n)|\leq C(\lambda,\alpha,\theta,\varepsilon) e^{-(\ln \lambda-\delta-\varepsilon) |n|}.
\end{equation}
By \eqref{gbo6new} and  Theorem \ref{edl}, one has that
for all $s$, $n$ and $\ell$,
\begin{equation}\label{gedlnewo}
	|\phi_s(n)\phi_s(\ell)|\leq C(\lambda,\alpha,\theta)e^{-C_1|n-\ell|+C_2|n|}.
	\end{equation}
 
From this, obtaining EDL is by now standard. In particular,  \eqref{gedlnewo} is precisely the input of Theorem 6.5 of \cite{tchecmp01}, with the output being EDL. 
\end{proof}
\begin{proof}[\bf Proof of the SULE  statement in part (2) of Theorem \ref{amcrit}] 

We are going to prove SULE when $\delta(\alpha,\theta)=0$
and $\lambda>1$.
Let $b>0$ be arbitrary small. In this proof, the largeness of  $C$ depends on $b$. 
Fix  any $\phi_s$.  If $|n|\geq C|m_s|$,  \eqref{gedl1} holds.
Therefore, we only need to consider  $|n|\leq C|m_s|$.

For any $n$ with $|n-m_s|\geq b|m_s|$, one has
\begin{equation*}
\min_{|x|\leq  C|m_s|}
|\sin\pi(2\theta+\alpha  x)|\geq  e^{-\frac{b}{2}|n-m_s|}.
\end{equation*}
By Theorem \ref{Keytheorem}, one has that for any $n$ with  $|n-m_s|\geq  b |m_s|$, 
\begin{equation*}
|\phi_s(n) |\leq e^{-  (\ln\lambda -b )|n-m_s|}.
\end{equation*}
It implies that  for any large $n\in \Z$,
\begin{equation*}
|\phi_s(n) |\leq  e^{b|m_s|-  (\ln\lambda-b)|n-m_s|}.
\end{equation*}
\end{proof}
 
 \section*{Acknowledgments}
%The authors are very grateful to the anonymous referees for their knowledgeable reports, which helped us to improve our manuscript.

W. Liu was a 2024-2025 Simons fellow.
 SJ’s
work was supported by NSF DMS-2052899, DMS-2155211, and Simons 896624. WL's work was supported   by NSF DMS-2246031,  DMS-2052572 and  DMS-2000345. LM's work was supported by  NNSFC-12371241 and  NSFSP-ZR2023MA032.
L. Mi and W. Liu would like to thank the math department at UC Irvine where this work was started, and W. Liu would like to thank the hospitality of the Department of Mathematics at UC Berkeley where this work was finished during his visit.

\end{document}